\documentclass[12pt]{article}
\usepackage[T1]{fontenc} 
\usepackage[utf8]{inputenx}
\usepackage{babel} 
\usepackage{url}
\usepackage{geometry}
\usepackage{amsmath}
\usepackage{amsfonts}
\usepackage{amstext}
\usepackage{amssymb}
\usepackage{amsthm}
\usepackage{amscd}
\usepackage{mathrsfs}
\usepackage{xcolor}
\definecolor{myurlcolor}{rgb}{0,0,0.4}
\definecolor{mycitecolor}{rgb}{0,0.5,0}
\definecolor{myrefcolor}{rgb}{0.5,0,0}
\usepackage{graphicx}
\usepackage{tikz}
\usepackage{tikz-cd}
\usepackage{etaremune}
\usepackage{latexsym}
\usepackage{braket}
\usepackage{soul}
\setstcolor{blue}
\usepackage{comment}

\usepackage{etoolbox}
\usepackage{makeidx}
\usepackage{dsfont}
\usepackage{enumitem} 
\usepackage{caption}
\usepackage{lmodern}
\usepackage{textcomp}
\usepackage{totcount}
\usepackage{blindtext}

\usepackage[pagebackref,draft=false]{hyperref}
\hypersetup{colorlinks,linkcolor=myrefcolor,citecolor=mycitecolor,urlcolor=myurlcolor}

\usepackage[capitalize]{cleveref}

\newtheorem{remark}{Remark}

\newtheorem{proposition}{Proposition}

\newtheorem*{proof*}{Proof}

\newcommand{\be}{\begin{equation}}
\newcommand{\ee}{\end{equation}}
\newcommand{\vsp}{\vspace{0.4cm}}
\newcommand{\grit}[1]{{\bfseries {\itshape {#1}}}}

\newcommand{\cfr}[1]{({\itshape cf.} {#1})}

\newcommand{\ra}{\rightarrow}

\newcommand{\hh}{\mathcal{H}}
\newcommand{\bh}{\mathcal{B}(\mathcal{H})}
\newcommand{\bhsa}{\mathcal{B}_{sa}(\mathcal{H})}
\newcommand{\Glh}{\mathcal{GL}(\mathcal{H})}

\newcommand{\Uh}{\mathcal{U}(\mathcal{H})}

\newcommand{\stsph}{\mathscr{S}(\mathcal{H})}
\newcommand{\posh}{\mathscr{P}(\mathcal{H})}
\newcommand{\Tr}{\textit{Tr}}

\newcommand{\gapp}{\mathscr{G}}

\title{G-dual teleparallel connections in Information Geometry}

\author{F. M. Ciaglia$^{1,5}$ \href{https://orcid.org/0000-0002-8987-1181}{\includegraphics[scale=0.7]{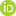}}, F. Di Cosmo$^{1,2,6}$ \href{https://orcid.org/0000-0003-0256-5913}{\includegraphics[scale=0.7]{ORCID.png}},  A. Ibort$^{1,2,7}$ \href{https://orcid.org/0000-0002-0580-5858}{\includegraphics[scale=0.7]{ORCID.png}},  G. Marmo$^{3,4,8}$ \href{https://orcid.org/0000-0003-2662-2193}{\includegraphics[scale=0.7]{ORCID.png}}
}

\begin{document}

\maketitle 

\noindent
{\footnotesize $^{1}$  Department of Mathematics, University Carlos III de Madrid, Legan\'es, Madrid, Spain   \\
$^{2}$ ICMAT, Instituto de Ciencias Matem\'{a}ticas (CSIC-UAM-UC3M-UCM)  \\
$^{3}$ INFN-Sezione di Napoli, Naples, Italy  \\
$^{4}$ Department of Physics ``E. Pancini'', University of Naples Federico II,  Naples, Italy

\bigskip
\noindent
$^{5}$\texttt{fciaglia[at]math.uc3m.es}, $^{6}$\texttt{fcosmo[at]math.uc3m.es}, $^{7}$\texttt{albertoi[at]math.uc3m.es}, $^{8}$\texttt{marmo[at]na.infn.it}}
	
\begin{abstract}

Given a real, finite-dimensional, smooth parallelizable Riemannian manifold $(\mathcal{N},G)$ endowed with a teleparallel connection $\nabla$ determined by a choice of a global basis of vector fields on $\mathcal{N}$, we show that the $G$-dual connection $\nabla^{*}$ of $\nabla$ in the sense of Information Geometry  must be the teleparallel connection determined by the basis of $G$-gradient vector fields associated with a basis of differential one-forms which is (almost) dual to the basis of vector fields determining $\nabla$.
We call any such pair $(\nabla,\nabla^{*})$ a \grit{$G$-dual teleparallel pair}.
Then, after defining a covariant $(0,3)$ tensor $T$ uniquely determined by   $(\mathcal{N},G,\nabla,\nabla^{*})$, we show that  $T$ being symmetric in the first two entries is equivalent to $\nabla$ being torsion-free, that $T$ being symmetric in the first and third entry is equivalent to $\nabla^{*}$ being torsion free,  and  that $T$ being  symmetric in the second and third entries is equivalent to the basis vectors determining $\nabla$ ($\nabla^{*}$)  being parallel-transported by $\nabla^{*}$ ($\nabla$).
Therefore, $G$-dual teleparallel pairs provide a generalization of the notion of Statistical Manifolds usually employed in Information Geometry, and we present explicit examples of  $G$-dual teleparallel pairs arising both in the context of both Classical and Quantum Information Geometry.

\end{abstract}

\tableofcontents

\section{Introduction}

From a purely mathematical point of view, the structure theory of Classical Information Geometry is the theory of Riemannian smooth manifolds with an additional structure.
This additional structure may be described, equivalently, using a pair of affine connections  or a covariant  $(0,3)$  tensor \cfr{\cite{Amari-1985,A-N-2000,Lauritzen-1987}}.
Specifically, in the first case we have a Riemannian smooth manifold $(\mathcal{N}, G)$, whose points parametrize probability distributions on a suitable outcome space, a torsion-free affine connection $\nabla$, and another torsion-free affine connection $\nabla^{*}$ which is dual to $\nabla$ with respect to $G$ in the sense that
\be\label{eqn: dual connections 0}
X(G(Y,Z))=G\left(\nabla_{X}Y,Z\right) + G\left(Y, \nabla^{*}_{X}Z\right)
\ee
for all vector fields $X,Y,Z$ on $\mathcal{N}$.
In general, $\nabla\neq \nabla^{*}$ unless $\nabla$ is the Levi-Civita connection $\nabla^{G}$ of $G$.
However, both $\nabla$ and $\nabla^{*}$ can be expressed in terms of $\nabla^{G}$ and a covariant tensor  essentially because the space of affine connections is, quite appropriately, an affine space modelled on a vector space of tensor fields.
Specifically, it holds
\be\label{eqn: 3-tensor 00}
G(\nabla_{X}Y,Z) - G(\nabla^{G}_{X}Y,Z) = T(X,Y,Z)
\ee
with $T$ a covariant $(0,3)$-tensor field on $\mathcal{N}$ which may also be written  according to 
\be\label{eqn: 3-tensor 11}
G(Y,\nabla_{X}^{G}Z) - G(Y,\nabla^{*}_{X}Z) =  T(X,Y,Z).
\ee
because of equation \eqref{eqn: dual connections 0} and because  $\nabla^{G}$ is its own $g$-dual connection.
If both $\nabla$ and $\nabla^{*}$ are torsion-free, it can be proved that $T$ is completely symmetric.

If instead of the pair $(\nabla,\nabla^{*})$ of torsion-free, $G$-dual affine connections we endow the Riemannian smooth manifold $(\mathcal{N},G)$ with a covariant $(0,3)$-tensor field $T$ which is totally symmetric, we can again obtain a pair of torsion-free, $G$-dual affine connections simply reading equation \eqref{eqn: 3-tensor 00} and \eqref{eqn: 3-tensor 11}  from right to left.
Therefore, we conclude that the two approaches are equivalent.
The triple $(\mathcal{N},G,T)$ is referred to as a \grit{statistical manifold} \cfr{\cite{Lauritzen-1987}}, and the tensor $T$ is often called \grit{Amari-\v{C}encov} tensor in honour of the two pioneers who thoroughly studied this structure \cfr{\cite{Amari-1985,A-J-L-S-2017,Cencov-1982}}.

It is worth noting how recently the theory of statistical manifolds has been approached from the point of view of Lie groupoids \cfr{\cite{G-G-K-M-2019,G-G-K-M-2020,G-G-K-M-2023}} where it is shown how Information Geometry on statistical manifolds is but a particular example of Information Geometry on Lie groupoids.
This point of view on Information Geometry may be particularly fruitful in connection with the recent reformulation of Quantum Theories in the framework of groupoids \cfr{\cite{C-DC-I-M-2020-02,C-I-M-2018,C-I-M-2019-02}}, especially considering how this reformulation of Quantum theories heavily relies on a categorical background which shares many interesting similarities with the one exploited by \v{C}encov in his pioneering work \cfr{\cite{Cencov-1982}}.
However, we only plan to investigate these matters in future works.

Coming back to statistical manifolds, of particular interest are those statistical manifolds whose Amari-\v{C}encov tensor $T$ gives rise to a pair $(\nabla,\nabla^{*})$ of $G$-dual connections which are both flat in the sense that, in addition to being torsion-free, their curvature vanishes.
In this case, the statistical manifold is called \grit{dually flat}  and some interesting properties of both applied and purely mathematical flavour present themselves \cfr{\cite{Amari-2016,A-N-2000,Shima-2007}}.
Of course, the request of flatness for the connections immediately imposes topological obstructions on the type of smooth manifold we can consider, and these kind of considerations also apply to the case of Quantum Information Geometry of finite-dimensional systems \cfr{\cite{A-T-2002,A-T-2003,Fujiwara-2021,Jencova-2003-2}}.

A typical example of a dually flat statistical manifold is given by the interior of the n-simplex endowed with the Fisher-Rao metric tensor $G_{FR}$ and with the pair of $G_{FR}$-dual affine connections known as the \grit{mixture} and \grit{exponential connection}.
The mixture connection on the interior of the n-simplex has a kind of privileged role because it is intimately connected with the convex structure of the set of probability distributions.
We refer to section \ref{sec: classical case} for a thorough discussion of this example.

It is well-known that a statistical manifold $(\mathcal{N},G,T)$ can be built suitably expanding up to third order a smooth function $S\colon\mathcal{N}\times\mathcal{N}\ra\mathbb{R}$ (also known as a two-point function on $\mathcal{N}$) satisfying some specific conditions  \cfr{\cite{A-A-2015,C-DC-L-M-M-V-V-2018,Matumoto-1993}}.
It often happens that the function $S$ is a relative entropy function, or a divergence function.
The most relevant family of relative entropies in the classical setting is that investigated by Csizar \cfr{\cite{Csizar-1963}}, and it turns out that every member in this family gives rise to the same Riemannian metric tensor which is precisely the Fisher-Rao metric tensor.

In the (finite-dimensional) quantum setting, it turns out that the situation is essentially different than in the classical setting.
First of all, on the manifold of faithful quantum states $\stsph$ of a physical system described by a finite-dimensional complex Hilbert space $\hh$ \cfr{section \ref{sec: quantum case}} there is an infinite number of Riemannian metric tensors whose theoretical (and often also applied) relevance has been thoroughly corroborated.
These Riemannian metric tensors are known as \grit{quantum monotone metric tensors} and have been completely classified by Petz \cfr{\cite{Petz-1996}}.
Every quantum monotone metric tensor may be obtained suitably expanding a quantum relative entropy which, in some sense, generalizes the classical relative entropy of Csizar mentioned above \cfr{\cite{C-DC-DN-V-2022,L-R-1999}}.
These instances are in stark contrast with what happens in the classical case where the Riemannian structure of essentially all the statistical manifolds considered  are  given by the Fisher-Rao metric tensor\footnote{We are here  deliberately ignoring all the Wasserstein-type metric tensors not because we believe they are not relevant but only because their very definition depend on additional structures  (e.g., a metric distance on the outcome space in the classical case or a metric distance on the space of pure states in the quantum case) that are not of interest for us here.}, and where every relative  entropy in the family investigated by Csizar leads to the Fisher-Rao metric tensor.

Of course, the dual pairs of affine connections generated by the quantum relative entropies associated with the quantum monotone metric tensors are necessarily torsion-free.
In particular, it is possible to build a mixture connection on $\stsph$ because, analogously to the space of classical probability distributions, the space of quantum states is a convex set \cfr{section \ref{sec: quantum case}}, and it turns out that there is only one monotone metric tensor - the so-called Bogoliubov-Kubo-Mori metric tensor \cfr{\cite{N-V-W-1975,Petz-1994,P-T-1993}} -  for which the mixture connection admits a dual connection which is the quantum analogue of the exponential connection and which is torsion-free \cfr{\cite{G-R-2001,Nagaoka-1995}}.
In this case, the associated quantum relative entropy is the von Neumann-Umegaki relative entropy \cfr{\cite{F-M-A-2019,Umegaki-1962-4,von-Neumann-1955}}.
Incidentally, the resulting statistical manifold is dually flat.

However, it turns out that there are pairs of affine connections which are dual with respect to the quantum monotone metric tensors and are not both torsion-free \cfr{\cite{A-N-2000,Fujiwara-1999,Jencova-2001}}.
In particular, examples of these dual pair with non-vanishing torsion are found when generalizing Amari's $\alpha$-representation \cfr{\cite{Amari-1985}} to the quantum case \cfr{\cite{Jencova-2001}}.
Moreover, one of the results of this work is to show that the affine connection which is dual to the mixture connection on $\stsph$ with respect to a specific quantum monotone metric tensor always has non-vanishing torsion unless the gradient vector fields of the expectation value functions of quantum mechanical observables commute.
Since the mixture connection is intimately connected with the convex structure of the space of quantum states and thus reflects a structural feature of $\stsph$ that goes beyond the formalism of Information Geometry, the idea of discarding all the cases in which the affine connection dual to the mixture one presents torsion seems too restrictive.
Indeed, already Amari and Nagaoka in their seminal work \cite[p. 19]{A-N-2000} noted how “the incorporation of torsion into the framework of information geometry, which would relate it to such fields as quantum mechanics (noncommutative probability theory) and systems theory, is an interesting topic for the future.”
Consequently, it seems that the notion of statistical manifold is not enough to completely understand the mathematical structures arising in Quantum Information Geometry.

Moreover, it is worth noting how the investigation of the appearance of torsion also in Classical Information Geometry is recently gaining  interest \cfr{\cite{H-M-2019,Z-K-2019,Z-K-2019-02,Z-K-2020}}.
In particular, the present work may be thought of as  complementing some theoretical aspects already discussed in \cite{Z-K-2019} (especially in connection with the explicit construction of what is called the g-biorthogonal frame), and providing a natural extension to the case of Quantum Information Geometry.

The aim of this work is to introduce a {\itshape variation on the theme} of dually flat connections  that allows the appearance of torsion and to show that this type of geometric structure naturally appears in Classical and Quantum Information Geometry.
The starting point is the idea that, at least locally, a flat connection $\nabla$ on a smooth manifold $\mathcal{N}$ admits a basis of vector fields $\{X_{j}\}_{j=1,...,n=\mathrm{dim}(\mathcal{N})}$ such that $\nabla_{X_{j}}X_{k}=0$ for all $j,k$.
If the vector fields are globally defined so that $\{X_{j}\}_{j=1,...,n=\mathrm{dim}(\mathcal{N})}$ is a basis for the module of vector fields on $\mathcal{N}$ (note that this requires $\mathcal{N}$ to be parallelizable), then the condition $\nabla_{X_{j}}X_{k}=0$ is what makes $\nabla$ be the so-called teleparallel connection associated with $\{X_{j}\}_{j=1,...,n=\mathrm{dim}(\mathcal{N})}$.
Moreover, a teleparallel connection associated with a basis $\{X_{j}\}_{j=1,...,n=\mathrm{dim}(\mathcal{N})}$  of vector fields is uniquely determined by the condition $\nabla_{X_{j}}X_{k}=0$ for all $j,k$.
In general, however, teleparallel connections need not be flat (i.e., curvature-free), nor even torsion-free (indeed, they are torsion-free if and only if the basis vector fields mutually commute as it can be easily seen from the very definition of the torsion tensor).
In almost all the cases arising from Classical and Quantum Information Geometry, it turns out that there is one preferred teleparallel connection which is basically the mixture connection arising from the convex structure characteristic of both the space of probability distributions and the space of quantum states.

Motivated by the previous discussion, in section \ref{sec: dually-weitzenbock manifolds}, given a fixed teleparallel connection on a smooth Riemannian manifold $(\mathcal{N},G)$,  we investigate the properties of the $G$-dual connection $\nabla^{*}$ and show that it must necessarily be itself a teleparallel connection.
We call any such pair $(\nabla,\nabla^{*})$ a \grit{$G$-dual teleparallel pair}.
Moreover, we also investigate the symmetry properties of the tensor $T$ generalizing the Amari-\v{C}encov tensor to the case of $G$-dual teleparallel pairs.
Then, in section \ref{sec: classical case} we review the case of the mixture and exponential connection on the interior $\Delta_{n}$ of the n-simplex highlighting the teleparallel structure of both connections, and showing how the teleparallel structure of the exponential connection is connected with a particular homogeneous manifold structure for $\Delta_{n}$.
In section \ref{sec: quantum case} we investigate the quantum case and characterize the $G_{f}$-dual teleparallel pairs $(\nabla,\nabla^{*})$ on the manifold $\stsph$ of faithful quantum states in finite dimensions endowed with a monotone quantum metric tensor $G_{f}$ in which $\nabla$ is   the mixture connection reflecting the convex structure of $\stsph$.
Finally, in section \ref{sec: conclusions} we gather some concluding remarks and future perspectives.

\section{G-dual teleparallel pairs}\label{sec: dually-weitzenbock manifolds}

Let $\mathcal{N}$ be a parallelizable smooth manifold and denote with $\mathfrak{X}(\mathcal{N})$ the module of smooth vector fields on $\mathcal{N}$.
Let $\{X_{j}\}_{j=1,...,n}$, with $n=\mathrm{dim}(\mathcal{N})$, be a global basis of vector fields on $\mathcal{N}$.
Define the \textit{teleparallel connection}\footnote{The term \textit{teleparallel}   my be translated with 'distantly parallel', and its use is mainly due to the fact that the affine connection we are interested in is used in the context of the so-called \textit{Teleparallel Gravity} \cite{A-P-2013}.
Note, however, that the teleparallel connection is also known as the \textbf{Weitzenböck connection} of $\{X_{j}\}_{j=1,...,n}$.} $\nabla$ associated with $\{X_{j}\}_{j=1,...,n}$ by setting
\be\label{eqn: weitzenbock connection}
\nabla_{X_{j}}X_{k}:=0 \quad \forall\;j,k=1,...,n. 
\ee
It is clear that every integral curve of $X_{j}$ is a geodesic of $\nabla$ for every $j=1,...,n$.
Moreover, since geodesics with fixed initial conditions are unique, every geodesic of $\nabla$ starting at $x\in\mathcal{N}$ with initial velocity $v_{x}\in T_{x}\mathcal{N}$ can be realized as an integral curve starting at $x$ of the vector field $X=a^{j}X_{j}$ with $a^{j}\in\mathbb{R}$ for all $j=1,...,n$ such that $v_{x}=X(x)$.

The torsion tensor $T^{\nabla}$ of $\nabla$, in general, reads
\be\label{eqn: torsion tensor}
T^{\nabla}(Z,W):=\nabla_{Z}W - \nabla_{W}Z - [Z,W]
\ee
for every $Z,W\in\mathfrak{X}(\mathcal{N})$.
Therefore, equation \eqref{eqn: weitzenbock connection} implies
\be\label{eqn: torsion teleparallel connection}
T^{\nabla}(X_{j},X_{k})=  [X_{k},X_{j}]
\ee
for every $j,k=1,...,n$, which means that, from the perspective of the basis $\{X_{j}\}_{j=1,...,n}$, the torsion tensor is connected with the commutator of the basis vector fields.
Of course, from equation \eqref{eqn: torsion teleparallel connection} it follows that $\nabla$ is torsion-free if and only if the basis vector fields mutually commute.

The curvature tensor $R^{\nabla}$ of $\nabla$ reads
\be
\left(R^{\nabla}(Z,W)\right)(V):=\nabla_{Z}\left(\nabla_{W} V\right) -\nabla_{W}\left(\nabla_{Z} V\right) - \nabla_{[Z,W]}V
\ee
for every $Z,W,V\in\mathfrak{X}(\mathcal{N})$.
Therefore, again equation \eqref{eqn: weitzenbock connection} implies
\be
\left(R^{\nabla}(X_{j},X_{k})\right)(X_{l}):= - \nabla_{[X_{j},X_{k}]}X_{l}=0
\ee
for every $j,k,l=1,...,n$, because $\{X_{j}\}_{j=1,...,n}$ is a basis of vector fields and thus $[X_{j},X_{k}]=F^{l}X_{l}$ with some smooth functions $F^{l}$ globally defined on $\mathcal{N}$.
We thus conclude that if  $\nabla$ is torsion-free  then it is necessarily flat (in the sense of being free of curvature and torsion).

\vsp

Let $G$ be a smooth Riemannian metric tensor on $\mathcal{N}$.
Motivated by the theory of Statistical Manifolds,  we want to investigate the pairs $(\nabla,\nabla^{*})$ of $G$-dual connections on $\mathcal{N}$ for which both $\nabla$ and $\nabla^{*}$ are teleparallel connections on $\mathcal{N}$.
We call any such pair $(\nabla,\nabla^{*})$ a \grit{$G$-dual teleparallel pair}.

At this purpose, it is instrumental that  we understand how $\nabla$ behaves on differential one-forms.
First of all, we consider a basis $\{\theta^{j}\}_{j=1,...,n}$ of  globally-defined differential  one-forms on $\mathcal{N}$ such that
\be\label{eqn: dual basis}
\theta^{j}(X_{k})=C^{j}_{k} .
\ee
where $C^{j}_{k}$ is a number for every $j,k=1,...,n$.
We call any such basis \grit{almost dual} to $\{X_{j}\}_{j=1,...,n}$.
In particular, when $C^{j}_{k}=\delta^{j}_{k}$ we obtain precisely the dual basis of  $\{X_{j}\}_{j=1,...,n}$.
We may think of $C_{jk}$ as a scaled version of $\delta_{jk}$ where the value at $j=k$ actually depends on $j$ and it is not necessarily always 1.
The fact that all the constructions presented in the work are valid for $C_{jk}$ instead of just $\delta_{jk}$ simply reflects the fact that the choice of each single vector field in the basis is determined up to a constant.

According to the general theory of affine connections \cfr{\cite{Besse-1987}}, the behaviour of $\nabla$ on the almost dual basis $\{\theta^{j}\}_{j=1,...,n}$ is encoded in the equation 
\be\label{eqn: weitzenbock connection on dual basis}
\left(\nabla_{Z}\theta^{j}\right)(W) = \nabla_{Z}\left(\theta^{j}(W)\right) -  \theta^{j}\left(\nabla_{Z}W\right)
\ee
where $Z,W$ are arbitrary vector fields on $\mathcal{N}$.
In particular, when $Z=X_{k}$ and $W=X_{l}$, equation \eqref{eqn: weitzenbock connection on dual basis} becomes
\be\label{eqn: weitzenbock connection on dual basis 2}
\left(\nabla_{X_{k}}\theta^{j}\right)(X_{l}) =   0 \quad \forall\;j,k,l=1,...,n
\ee
because of equation \eqref{eqn: weitzenbock connection}, and because the covariant derivatives of the $C_{jk}$'s in equation \eqref{eqn: dual basis} vanish being the $C_{jk}$'s constant.

Equation \eqref{eqn: weitzenbock connection on dual basis 2} is equivalent to
\be\label{eqn: weitzenbock connection on dual basis 4}
\nabla_{Z}\theta^{j}  =   0 \quad \forall\;j  =1,...,n \quad \;\forall \;Z\in\mathfrak{X}(\mathcal{N})
\ee
because $\{X_{j}\}_{j=1,...,n}$ is a basis for $\mathfrak{X}(\mathcal{N})$ and, for every smooth function $f$ on $\mathcal{N}$, it holds $\nabla_{fZ}\alpha=f\nabla_{Z}\alpha$ for every differential one-form $\alpha$ on $\mathcal{N}$. 

Since the Riemannian metric tensor $G$ is non-degenerate, for every $\theta^{j}$ we can define its associated $G$-gradient (or simply 'gradient' if there is no risk of confusion) vector field $Y_{j}$ uniquely determined by
\be\label{eqn: gradient vector field}
\theta^{j}(Z) = G(Y_{j},Z)
\ee
for every $Z\in\mathfrak{X}(\mathcal{N})$.
Note that $\{Y_{j}\}_{j=1,...,n}$ is a basis of globally-defined vector fields on $\mathcal{N}$.
Setting $Z=Y_{k}$ and $W=X_{l}$ and exploiting equation \eqref{eqn: weitzenbock connection on dual basis 4}, it follows that equation \eqref{eqn: weitzenbock connection on dual basis} becomes
\be\label{eqn: dually weitzenbock 1}
0= \theta^{j}(\nabla_{Y_{k}}X_{l}) 
\ee
for all $j,k,l=1,...,n,$ which is equivalent to
\be\label{eqn: dually weitzenbock 2}
\nabla_{Y_{k}}X_{l}=0  
\ee
for all $k,l=1,...,n$.

With these elements at our disposal, we now characterize the affine connection $\nabla^{*}$ which is $G$-dual to the teleparallel connection $\nabla$.
In particular, it is proved that $\nabla^{*}$ is the teleparallel connection associated with $\{Y_{j}\}_{j=1,...,n}$.
It is worth noting that the content of proposition \ref{prop: teleparallel dual connection} is essentially the same as that of Theorem 10 in \cite{Z-K-2019}.

\begin{proposition}\label{prop: teleparallel dual connection}
	Let $\mathcal{N}$ be a parallelizable smooth manifold of dimension $n=\mathrm{dim}(\mathcal{N})$.
	Let   $\{X_{j}\}_{j=1,...,n}$ be the global basis of vector fields on $\mathcal{N}$, let  $\nabla$ be its associated teleparallel connection, and let $\{\theta^{j}\}_{j=1,...,n}$ be an almost dual basis of differential  one-forms  for $\{X_{j}\}_{j=1,...,n}$ \cfr{equation \eqref{eqn: dual basis}}.
	For any smooth Riemannian metric tensor  $G$ on $\mathcal{N}$, let $\{Y_{j}\}_{j=1,...,n}$ be the global basis of vector fields on $\mathcal{N}$ made up of $G$-gradient vector fields of the $\theta^{j}$'s \cfr{equation \eqref{eqn: gradient vector field}}.
	Then, the $G$-dual affine connection $\nabla^{*}$ of $\nabla$ \cfr{equation \eqref{eqn: dual connections 0}} is the teleparallel connection associated with $\{Y_{j}\}_{j=1,...,n}$.
	
\end{proposition}

\begin{proof}
	
	The $G$-dual connection $\nabla^{*}$ is uniquely characterized by   
	\be\label{eqn: dual connections}
	Z\left(G(W,V)\right)=G\left(\nabla_{Z}W,V\right) + G\left(W,\nabla_{Z}^{*}V\right)
	\ee
	for every $Z,W,V\in\mathfrak{X}(\mathcal{N})$.
	Note that 
	\be\label{eqn: dually weitzenbock 3}
	Y_{j}\left(G(Y_{k},X_{l})\right)\stackrel{\mbox{\eqref{eqn: gradient vector field}}}{=}Y_{j}\left(\theta^{k}(X_{l})\right)\stackrel{\mbox{\eqref{eqn: dual basis}}}{=}Y_{j}(C^{k}_{l})=0 .
	\ee
	Therefore, setting $Z=Y_{j}$, $W=X_{k}$, and $V=Y_{l}$, equation \eqref{eqn: dual connections} becomes
	\be \label{eqn: dually weitzenbock 4}
	0\stackrel{\mbox{\eqref{eqn: dually weitzenbock 3}}}{=}Y_{j}\left(G(X_{k},Y_{l})\right)=G\left(\nabla_{Y_{j}}X_{k},Y_{l}\right) + G\left(X_{k},\nabla_{Y_{j}}^{*}Y_{l}\right)\stackrel{\mbox{\eqref{eqn: dually weitzenbock 2}}}{=}  G(X_{k},\nabla_{Y_{j}}^{*}Y_{ l})  
	\ee
	for all $j,k,l=1,...,n$, which is equivalent to
	\be\label{eqn: dually weitzenbock 5}
	\nabla_{Y_{j}}^{*}Y_{ l} =0 
	\ee
	for all $j,l=1,...,n$ because $\{X_{j}\}_{j=1,...,n}$ is a basis of vector fields on $\mathcal{N}$.
	We thus conclude that the $G$-dual connection $\nabla^{*}$   is precisely the teleparallel connection determined by the basis $\{Y_{j}\}_{j=1,...,n}$ of gradient vector fields on $\mathcal{N}$ associated with the almost dual basis $\{\theta^{j}\}_{j=1,...,n}$  of    $\{X_{j}\}_{j=1,...,n}$ through $G$ as claimed.
	\qed 
\end{proof}

\begin{remark}\label{rem: dually dual teleparallel connection}
	We can introduce the basis $\{\alpha^{j}\}_{j=1,...,n}$ of  differential one-forms where each $\alpha^{j}$ is the $G$-dual of $X_{j}$ according to
	\be\label{eqn: gradient vector field 2}
	\alpha^{j}(Z)= G(X_{j},Z)
	\ee
	for every $Z\in\mathfrak{X}(\mathcal{N})$.
	Then, it follows that $\{\alpha^{j}\}_{j=1,...,n}$ is an almost dual basis of $\{Y_{j}\}_{j=1,...,n}$ because
	\be\label{eqn: dually dual basis}
	\alpha^{j}(Y_{k})\stackrel{\mbox{\eqref{eqn: gradient vector field 2}}}{=}G(X_{j},Y_{k})\stackrel{\mbox{\eqref{eqn: gradient vector field}}}{=}\theta^{k}(X_{j})\stackrel{\mbox{\eqref{eqn: dual basis}}}{=}C^{k}_{j}.
	\ee
	Consequently, proceeding in perfect analogy with what is done to obtain equation \eqref{eqn: dually weitzenbock 2}, we obtain
	\be\label{eqn: dually dual weitzenbock}
	\nabla_{X_{k}}^{*}Y_{l}=0  
	\ee
	for all $k,l=1,...,n$.
\end{remark}

\begin{remark}\label{rem: closed differentials and vanishing torsion}
	Recalling the formula
	\be
	\mathrm{d}\beta(Z,W)=Z(\beta(W)) - W(\beta(Z)) - \beta([Z,W])
	\ee
	valid for every differential one-form $\beta$ and every pair of vector fields $Z,W$ on $\mathcal{N}$, we immediately see that
	\be\label{eqn: differential and commutators 1}
	\begin{split}
		\mathrm{d}\theta^{j}(X_{k},X_{l})&=X_{k}(\theta^{j}(X_{l})) - X_{k}(\theta^{j}(X_{l})) -\theta^{j}([X_{k},X_{l}])\stackrel{\mbox{\eqref{eqn: dual basis}}}{=}-\theta^{j}([X_{k},X_{l}]),
	\end{split}
	\ee
	which means that  $\nabla$ is torsion-free \cfr{equation \eqref{eqn: torsion teleparallel connection}}  if and only if $\{\theta^{j}\}_{j=1,...,n}$ is a basis of closed differential one-forms.
	Similarly,    we have
	\be\label{eqn: differential and commutators 2}
	\begin{split}
		\mathrm{d}\alpha^{j}(Y_{k},Y_{l})&=Y_{k}(\alpha^{j}(Y_{l})) - Y_{l}(\alpha^{j}(Y_{k})) -\alpha^{j}([Y_{k},Y_{l}])\stackrel{\mbox{\eqref{eqn: dually dual basis}}}{=}-\alpha^{j}([Y_{k},Y_{l}])
	\end{split}
	\ee
	which means that  $\nabla^{*}$ is torsion-free \cfr{equation \eqref{eqn: torsion teleparallel connection}}  if and only if $\{\alpha^{j}\}_{j=1,...,n}$ is a basis of closed differential one-forms.

\end{remark}

If $(\nabla,\nabla^{*})$ is a $G$-dual  pair on $(\mathcal{N},G)$ and $\nabla^{G}$ is the Levi-Civita connection of $G$, we may introduce the covariant $(0,3)$ tensor $T$ setting
\be\label{eqn: 3-tensor 0}
T(Z,W,V):=G(\nabla_{Z}W,V) - G(\nabla^{G}_{Z}W,V) 
\ee
for every $Z,W,V\in\mathfrak{X}(\mathcal{N})$.
Since $\nabla^{*}$ is the $G$-dual connection of $\nabla$, equation \eqref{eqn: 3-tensor 0} can be re-written in terms of $\nabla^{*}$ as
\be\label{eqn: 3-tensor 1}
T(Z,W,V):= G(W,\nabla^{G}_{Z}V)  -G(W,\nabla_{Z}^{*}V) .
\ee
If both $\nabla$ and $\nabla^{*}$ are torsion-free, then $T$ is completely symmetric and it is precisely the Amari-\v{C}encov tensor of the Statistical Manifold $(\mathcal{N},G,T)$.
In general, however, $T$ is no longer completely symmetric and we now investigate its properties when $(\nabla,\nabla^{*})$ is  a $G$-dual teleparallel pair.

\begin{proposition}
	Let $(\nabla,\nabla^{*})$ be  a $G$-dual teleparallel pair on the real, finite-dimensional, smooth, parallelizable Riemannian manifold $(\mathcal{N},G)$ of dimension $\mathrm{dim}(\mathcal{N})=n$, and let $T$ be the covariant $(0,3)$ tensor defined in equation \eqref{eqn: 3-tensor 0} (or equation \eqref{eqn: 3-tensor 1}).
	Then:
	\begin{itemize}
		\item $T$ is symmetric in the first two entries if and only if  $\nabla$ is torsion-free  or, equivalently, if and only if $\{\theta^{j}\}_{j=1,...,n}$ is a basis of closed differential one-forms;
		\item  $T$ is symmetric in the first and third entries if and only if  $\nabla^{*}$ is torsion-free  or, equivalently, if and only if $\{\alpha^{j}\}_{j=1,...,n}$ is a basis of closed differential one-forms;
		\item $T$ is symmetric in the second and third entries if and only if $\nabla^{*}_{X_{j}}X_{k}=0$ for every $j,k=1,...,n$, or, equivalently, if and only if $\nabla_{Y_{j}}Y_{k}=0$ $j,k=1,...,n$.
	\end{itemize}
	
\end{proposition}
\begin{proof}
	In this case it holds
	\be
	T(X_{j},X_{k},Y_{l})\stackrel{\mbox{\eqref{eqn: 3-tensor 0}\eqref{eqn: weitzenbock connection}}}{=} - G(\nabla^{G}_{X_{j}}X_{k},Y_{l}).
	\ee
	Consequently, since $\nabla^{G}$ is torsion-free, we have $\nabla^{G}_{X_{j}}X_{k}-\nabla^{G}_{X_{k}}X_{j}=[X_{j},X_{k}]$ and thus
	\be
	T(X_{j},X_{k},Y_{l})-T(X_{k},X_{j},Y_{l})=G\left([X_{k},X_{j}],Y_{l}\right)\stackrel{\mbox{\eqref{eqn: gradient vector field}}}{=}-\theta^{l}([X_{j},X_{k}])\stackrel{\mbox{\eqref{eqn: differential and commutators 1}}}{=}\mathrm{d}\theta^{j}(X_{k},X_{l})
	\ee
	showing that $T$ is symmetric in the first two entries if and only if  $\nabla$ is torsion-free  or, equivalently, if and only if $\{\theta^{j}\}_{j=1,...,n}$ is a basis of closed differential one-forms \cfr{remark \ref{rem: closed differentials and vanishing torsion}} as claimed.
	In particular, this is always true when we consider the mixture connection in both the classical and quantum case \cfr{sections \ref{sec: classical case} and \ref{sec: quantum case}}.

	Analogously, we have
	\be
	T(Y_{j},X_{k},Y_{l})\stackrel{\mbox{\eqref{eqn: 3-tensor 1}\eqref{eqn: weitzenbock connection}}}{=}   G(X_{k},\nabla^{G}_{Y_{j}}Y_{l}),
	\ee
	so that, again since $\nabla^{G}$ is torsion-free, we can use $\nabla^{G}_{Y_{j}}Y_{l}-\nabla^{G}_{Y_{l}}Y_{j}=[Y_{j},Y_{l}]$ and obtain
	\be\label{eqn: T-symmetricity 2}
	\begin{split}
		T(Y_{j},X_{k},Y_{l})-T(Y_{l},X_{k},Y_{j})&= G(X_{k},[Y_{j},Y_{l}])\stackrel{\mbox{\eqref{eqn: gradient vector field 2}}}{=}-\alpha^{k}([Y_{j},Y_{l}])\stackrel{\mbox{\eqref{eqn: differential and commutators 2}}}{=}\mathrm{d}\alpha^{k}(Y_{j},Y_{l}),
	\end{split}
	\ee
	which means that $T$ is symmetric in the first and third entries if and only if  $\nabla^{*}$ is torsion-free  or, equivalently, if and only if $\{\alpha^{j}\}_{j=1,...,n}$ is a basis of closed differential one-forms \cfr{remark \ref{rem: closed differentials and vanishing torsion}}.
	
	Finally, recalling that
	\be\label{eqn: dual Levi-Civita connection}
	Z(G(W,V))=G(\nabla^{G}_{Z}W,V) + G(W,\nabla^{G}_{Z}V)
	\ee
	for all $Z,W,V\in \mathfrak{X}(\mathcal{N})$ because $\nabla^{G}$ is its own $G$-dual connection, we have
	\be\label{eqn: T-symmetricity 3}
	\begin{split}
		T(X_{j},X_{k},X_{l})-T(X_{j},X_{l},X_{k})&\stackrel{\mbox{\eqref{eqn: 3-tensor 1}}}{=}G\left(X_{k},\nabla^{G}_{X_{j}}X_{l}- \nabla^{*}_{X_{j}}X_{l}\right) - G\left(X_{l},\nabla^{G}_{X_{j}}X_{k}- \nabla^{*}_{X_{j}}X_{k} \right)=\\
		&\stackrel{\mbox{\eqref{eqn: dual connections 0}\eqref{eqn: weitzenbock connection}\eqref{eqn: dual Levi-Civita connection}}}{=}    G\left(X_{l},\nabla^{*}_{X_{j}}X_{k} \right),
	\end{split}
	\ee
	which means that $T$ is symmetric in the second and third entries if and only if $\nabla^{*}_{X_{j}}X_{k}=0$ for every $j,k=1,...,n$.
	Moreover, it also holds
	\be\label{eqn: T-symmetricity 4}
	\begin{split}
		T(Y_{j},Y_{k},Y_{l})-T(Y_{j},Y_{l},Y_{k})&\stackrel{\mbox{\eqref{eqn: 3-tensor 0}}}{=}G\left(\nabla_{Y_{j}}Y_{k}-\nabla^{G}_{Y_{j}}Y_{k},Y_{l}\right) - G\left(\nabla_{Y_{j}}Y_{l}-\nabla^{G}_{Y_{j}}Y_{l},Y_{k}\right)=\\
		&\stackrel{\mbox{\eqref{eqn: dual connections 0}\eqref{eqn: dually weitzenbock 5}\eqref{eqn: dual Levi-Civita connection}}}{=}  - G\left(\nabla_{Y_{j}}Y_{l} ,Y_{k}\right)
	\end{split}
	\ee
	which means that $T$ is symmetric in the second and third entries if and only if $\nabla_{Y_{j}}Y_{k}=0$ for every $j,k=1,...,n$.
	We thus conclude that 
	\be
	\nabla^{*}_{X_{j}}X_{k}=0\;\;\forall\;j,k=1,...,n \,\Longleftrightarrow\; \nabla_{Y_{j}}Y_{k}=0 \;\;\forall\;j,k=1,...,n.
	\ee
	Of course, if $\nabla$ and $\nabla^{*}$ are both torsion-free then it can easily be checked that $T$ is completely symmetric.
\end{proof}

\section{G-dual teleparallel pairs in Classical Information Geometry} \label{sec: classical case}


Let $\overline{\Delta_{n}}$  be the  n-simplex in $\mathbb{R}^{n}$:
\be
\overline{\Delta_{n}}:=\left\{\vec{p}=(p^{1},...,p^{n})\in\mathbb{R}^{n}\,|\quad p^{j}\geq 0\;\forall\; j=1,...,n,\mbox{ and } \sum_{j=1}^{n}p^{j}=1\right\},
\ee
and $\Delta_{n}$ its interior 
\be
\Delta_{n}:=\left\{\vec{p}=(p^{1},...,p^{n})\in\overline{\Delta_{n}}\,|\quad p^{j}>0\;\forall\; j=1,...,n \right\}.
\ee
Note that both $\overline{\Delta_{n}}$ and $\Delta_{n}$ are convex subsets of $\mathbb{R}^{n}$.
If $\mathcal{X}_{n}$ denote a discrete set with $n$ elements, then  $\overline{\Delta_{n}}$  can be interpreted as the  set of  probability distributions  on $\mathcal{X}_{n}$, and $\Delta_{n}$ as the subset of nowhere-vanishing probability distributions on $\mathcal{X}_{n}$.
The space $\Delta_{n}$ is a smooth manifold which is also  parallelizable   because it is essentially an open subset of the affine hyperplane determined by the normalization condition $\sum_{j=1}^{n}p^{j}=1$.

A tangent vector at $\vec{p}\in\Delta_{n}$   can be identified with an element $\vec{a}\in\mathbb{R}^{n}$ such that $\sum_{j=1}^{n}a^{j}=0$.
If we choose $(n-1)$ linearly independent vectors $\vec{a}_{1},...,\vec{a}_{n-1}\in\mathbb{R}^{n}$ such that $\sum_{j=1}^{n}a^{j}_{k}=0$ for all $k=1,...,(n-1)$, we  can define a basis $\{L_{k}\}_{k=1,...,(n-1)}$   of global vector fields on $\Delta_{n}$ setting  
\be\label{eqn: traslation vector fields on the simplex}
L_{k}(\vec{p})=\vec{a}_{k} .
\ee
The notational change from $X_{k}$ to $L_{k}$ is made to remark the relation of these vector fields with the linear structure of $\mathbb{R}^{n}$, and thus with the convex structures of probability vectors.
Indeed, every $L_{k}$ generates a traslation on $\mathbb{R}^{n}$, and traslations are related with the linear structure of $\mathbb{R}^{n}$.

The teleparallel connection associated with $\{L_{k}\}_{k=1,...,(n-1)}$  is known as the mixture connection on $\Delta_{n}$ and it is often denoted by $\nabla^{m}$.
This affine connection is intimately connected with the convex structure of $\Delta_{n}$.
Indeed, as mentioned in section \ref{sec: dually-weitzenbock manifolds}, every geodesic of $\nabla^{m}$ starting at $\vec{p}$ and with intial velocity $\vec{v}=v^{k}a_{k}=v^{k}L_{k}(\vec{p})$, where $v^{k}\in\mathbb{R}$ for all $k=1,...,(n-1)$, can be realized as an integral curve of the vector field $V=v^{k}L_{k}$, and has the simple expression 
\be
\vec{p}_{\vec{v}}(t)=\vec{p} + t\, \vec{v} .
\ee
Therefore, given two arbitrary points $\vec{p}_{1},\vec{p}_{2}\in\Delta_{n}$, we can build a geodesic starting at $\vec{p}_{1}$ and ending at $\vec{p}_{2}$ at time $t=1$ setting $\vec{v}=\vec{p}_{2}-\vec{p}_{1}$ thus obtaining the segment $(1-t)\vec{p}_{1} + t\vec{p}_{2}$ describing all the possible convex combinations of $\vec{p}_{1}$ and $\vec{p}_{2}$.

It is perhaps seldom observed that, from the geometrical point of view, $\Delta_{n}$ is actually an homogeneous space of the  multiplicative Lie group $\mathbb{R}^{n}_{+}$, i.e., the product of $n$ copies of $\mathbb{R}_{+}$.
Specifically, let $(q^{1},...,q^{n})=\vec{q}\in\mathbb{R}^{n}_{+}$ and set
\be\label{eqn: geodesics of classical exponential connection}
\gamma(\vec{q},\vec{p}):=\frac{1}{N_{\vec{q}}^{\vec{p}}}\,\left(q^{1}p^{1},...,q^{n}p^{n}\right)
\ee
with $N_{\vec{q}}^{\vec{p}}=\sum_{j=1}^{n}q^{j}p^{j}$.
It is a matter of direct computation to check that $\gamma$ is indeed a smooth, transitive left action of  $\mathbb{R}^{n}_{+}$ on $\Delta_{n}$.
Then, setting $\vec{b}_{0}=(1,...,1)$ and selecting  $(n-1)$ linearly independent vectors $\vec{b}_{1},...,\vec{b}_{n-1}\in\mathbb{R}^{n}$ such that $\sum_{j=1}^{n}b^{j}_{k}=0$ for all $k=1,...,(n-1)$, we obtain a basis $\{\vec{b}_{j}\}_{j=0,...,(n-1)}$ of the Lie algebra $\mathbb{R}^{n}$ of $\mathbb{R}^{n}$ such that the fundamental vector field of $\gamma$ associated with $\vec{b}_{0}$ identically vanishes, while the fundamental vector field of $\gamma$ associated with $\vec{b}_{j}$ reads
\be\label{eqn: gradient vector field e-connection n-simplex}
Y_{j}(\vec{p})=\left(\left(b^{1}_{j} - N_{\vec{b}_{j}}^{\vec{p}}\right)p^{1},...\,,\left(b^{n}_{j} - N_{\vec{b}_{j}}^{\vec{p}}\right)p^{n}\right) .
\ee
It turns out that the teleparallel connection determined by the basis $\{Y_{j}\}_{j=1,...,(n-1)}$ of vector fields on $\Delta_{n}$ is precisely the exponential connection which is dual to $\nabla^{m}$ with respect to the Fisher-Rao metric tensor.

According to \v{C}encov's pioneering work \cfr{\cite{Cencov-1982}}, the Riemannian manifold structure which is relevant in the context of Classical Information Geometry is the so-called Fisher-Rao metric tensor $G_{FR}$ given by
\be\label{eqn: F-R metric tensor}
\left(G_{FR}(X,Y)\right)(\vec{p}):=\langle X(\vec{p}),Y(\vec{p})\rangle_{\vec{p}}
\ee
where  $\langle \cdot , \cdot\rangle_{\vec{p}}$  is the weighted scalar product on $\mathbb{R}^{n}$ given by 
\be
\langle \vec{a},\vec{b}\rangle_{\vec{p}}:=\sum_{j=1}^{n}\frac{a^{j}b^{j}}{p^{j}}.
\ee
It is well-known that the $G_{FR}$-dual connection of $\nabla^{m}$ is the so-called exponential connection $\nabla^{e}$ \cfr{\cite{A-N-2000,Cencov-1982}}.
Moreover, from proposition \ref{prop: teleparallel dual connection} we know that $\nabla^{e}$ is the teleparallel connection associated with the gradient vector fields of an almost dual basis of  $\{L_{k}\}_{k=1,...,(n-1)}$.
At this purpose,  an almost dual basis $\{\theta^{k}\}_{k=1,...,(n-1)}$ of $\{L_{k}\}_{k=1,...,(n-1)}$ is associated with a choice of $(n-1)$ linearly independent vectors $\{\vec{b}_{k}\}_{k=1,...,(n-1)}$ such that $\sum_{j=1}^{n}b^{j}_{k}=0$ for all $k=1,...,(n-1)$.
Indeed, setting $\theta^{k}=\mathrm{d}l^{k}$ with 
\be
l^{k}(\vec{p}):=\vec{p}\cdot\vec{b}_{k} ,
\ee
it immediately follows that 
\be\label{eqn: dual basis simplex}
\theta^{k}(L_{j})=\frac{\mathrm{d}}{\mathrm{d}t}\left((\vec{p}+t\vec{a}_{j})\cdot\vec{b}_{k}\right)_{t=0}=\vec{a}_{j}\cdot\vec{b}_{k} ,
\ee
and the right hand side clearly does not depend on the point $\vec{p}$, thus showing that $\{\theta^{k}\}_{k=1,...,(n-1)}$ is indeed an almost dual basis of $\{L_{k}\}_{k=1,...,(n-1)}$.
Then, a direct computation using the very definition of gradient vector field \cfr{equation \eqref{eqn: gradient vector field}}, and equations \eqref{eqn: F-R metric tensor} and \eqref{eqn: dual basis simplex} shows that the gradient vector field associated with $\theta^{k}$   by the Fisher-Rao metric tensor $G_{FR}$  is precisely the vector field $Y_{k}$ given in equation  \eqref{eqn: gradient vector field e-connection n-simplex} as claimed.
Since the exponential connection $\nabla^{e}$ is the teleparallel connection determined by the basis $\{Y_{j}\}_{j=1,...,(n-1)}$, every geodesic of $\nabla^{e}$ starting at $\vec{p}$ and with intial velocity $\vec{v}=v^{j}a_{j}=v^{j}Y_{j}(\vec{p})$, where $v^{j}\in\mathbb{R}$ for all $j=1,...,(n-1)$, can be realized as an integral curve of the vector field $V=v^{j}Y_{j}$.
Then, since the $Y_{j}$'s are the fundamental vector fields of the action of $\mathbb{R}^{n}_{+}$ on $\Delta_{n}$ described in equation \eqref{eqn: geodesics of classical exponential connection}, it immediately follows that every integral curve of $V$ starting at $\vec{p}$ reads
\be\label{eqn: geodesics of classical exponential connection 2}
\vec{p}_{\vec{v}}(t)=\frac{1}{\sum_{j=1}^{n} \mathrm{e}^{tv^{j}}p^{j}}\,\left(\mathrm{e}^{tv^{1}}\,p^{1},...,\mathrm{e}^{tv^{n}}\,p^{n}\right),
\ee
from which it is evident that $\nabla^{e}$ is complete.

\section{G-dual teleparallel pairs in Quantum Information Geometry}\label{sec: quantum case}

The purpose of this section is to discuss the $G$-dual teleparallel pairs   that naturally arise in the context of finite-dimensional Quantum Information Geometry when the reference teleparallel connection comes from the natural convex structure of the space of quantum states, and the metric tensors are the so-called monotone metric tensors replacing the Fisher-Rao metric tensor in the quantum context and classified by Petz  \cite{Petz-1996}.

Let $\hh$ be a complex Hilbert space of dimension $n<\infty$.
Let  $\langle \cdot , \cdot \rangle$  denote the Hilbert space inner product on $\hh$, $\bh$ denote the space of bounded linear operators on $\hh$, and $\bhsa$ the space of self-adjoint elements in $\bh$.
The vector space $\bhsa$ is a real Hilbert space with respect to the so-called Hilbert-Schmidt product
\be
\langle \mathbf{a},\mathbf{b}\rangle_{HS}:=\Tr(\mathbf{a}\mathbf{b}),
\ee
where $\Tr$ is the Hilbert space trace.
We can always select a  basis $\{\sigma_{j}\}_{j=0,...,n^{2}-1}$ for $\bhsa$ in such a way that $\sigma_{0}=\mathbb{I}$, the identity operator on $\hh$, and  $\Tr(\sigma_{j})=0$ for all $j=1,...,n^{2}-1$.
When $\hh\cong\mathbb{C}^{2}$, the typical example of such a basis is the one in which  $\sigma_{j}$ is the j-th Pauli matrix for $j=1,2,3$.

Every element $\mathbf{a}\in\bhsa$ can then be expressed as
\be\label{eqn: expansion of s-a linear operators}
\mathbf{a}=x^{0}\sigma_{0} + \vec{x}\cdot\vec{\sigma}
\ee
where $\vec{x}=(x^{1},...,x^{n^{2}-1})\in\mathbb{R}^{n^{2}-1}$, $\vec{\sigma}=(\sigma^{1},...,\sigma^{n^{2}-1})$, and  $\cdot$ is the Euclidean scalar product.
The real numbers $(x^{0},\vec{x})$ determine a Cartesian coordinate system on $\bhsa$.

An element in $\bhsa$ is called positive  if   its spectrum lies on the non-negative half-line, and strictly positive when its spectrum lies on the positive half-line.
In particular, a strictly positive   operator is invertible.
It is customary to write $\rho\geq 0$ to denote a positive operator, and $\rho>0$ to denote a strictly positive one.
The space $\posh$ of strictly positive operators on $\hh$ is an open convex cone in $\bhsa$, which means it is an open submanifold of $\bhsa$, and thus a parallelizable smooth manifold since $\bhsa$ is a vector space.

The space of faithful quantum states of a finite-level quantum systems associated with $\hh$ is identified with the space of invertible density operators on $\hh$ given by
\be\label{eqn: stsph}
\stsph:=\left\{\rho\in\posh |\;\Tr(\rho)=1\right\}.
\ee
Accordingly, $\stsph$ is a codimension-one submanifold of $\posh$ which is also parallelizable  because it is essentially an open subset of the affine hyperplane in $\bhsa$ determined by the trace-normalization condition.
Moreover,  it is also a convex set.

The space of all quantum states $\overline{\stsph}$ is the norm-closure of $\stsph$ inside $\bh$ (or $\bhsa$), and it is a compact, convex subset of $\bhsa$ which is also a stratified manifold \cite{DA-F-2021,G-K-M-2005}.

Recalling equation \eqref{eqn: expansion of s-a linear operators} and the normalization condition in equation \eqref{eqn: stsph}, it is clear that every $\rho\in\overline{\stsph}$ is such that $x^{0}=\frac{1}{n}$.
However, it is not in general possible to explicitly find all the constraints the vector $\vec{x}$ in equation equation \eqref{eqn: expansion of s-a linear operators} must satisfy in order to ensure the positivity condition $\rho\geq 0$.
The two-dimensional case is the only case in which a direct computation of the eigenvalues shows that $\rho\geq 0$ is equivalent to $\vec{x}\cdot \vec{x}\leq\frac{1}{4}$.

Because of the normalization condition, and because of dimensional reasons related with the fact that $\stsph$ is a codimension-one submanifold of the open submanifold $\posh\subset\bhsa$, the tangent space at $\rho$ can be expressed as the linear span
\be
T_{\rho}\stsph \cong\mathrm{span}\left(\sigma_{1},...,\sigma_{n^{2}-1}\right) ,
\ee
and we can define the global basis $\{L_{k}\}_{k=1,...,n^{2}-1}$ of vector fields on $\stsph$ setting
\be\label{eqn: linear vector field quantum}
L_{k}(\rho)=\sigma_{k} .
\ee
Again, the notational change from $X_{k}$ to $L_{k}$ is made to emphasize the relation of these vector fields with the linear structure of $\bhsa$  \cfr{the comment after equation \eqref{eqn: traslation vector fields on the simplex}}.

The teleparallel connection associated with $\{L_{k}\}_{k=1,...,n^{2}-1}$  is known as the mixture connection on $\stsph$ and it is denoted as $\nabla^{m}$.
Very much like it happens in the classical setting for the mixture connection on the interior of the n-simplex, the mixture connection on $\stsph$ turns out to have a natural geometric interpretation in terms of the convex structure of $\stsph$.
Indeed, every geodesic of $\nabla^{m}$ starting at $\rho$ with initial tangent vector $\mathbf{v}=v^{k}\sigma_{k}=v^{k}L_{k}(\rho)$, where $v^{k}\in\mathbb{R}$ for every $k=1,...,(n^{2}-1)$,  can be realized as the integral curve of the vector field $V=v^{k}L_{k}$ \cfr{section \ref{sec: dually-weitzenbock manifolds}} which reads
\be\label{eqn: geodesic of mixture connection quantum case}
\rho_{\mathbf{v}}(t)= \rho + t\mathbf{v} .
\ee
Therefore, given two arbitrary points $\rho_{1},\rho_{2}\in\stsph$, we can build a geodesic starting at $\rho_{1}$ and ending at $\rho_{2}$ at time $t = 1$ by setting $\mathbf{v}=\rho_{2}-\rho_{1}$,  thus obtaining the segment $(1-t)\rho_{1} + t\rho_{2}$ describing all the possible convex combinations of  $\rho_{1}$ and $\rho_{2} $.
Moreover, it follows from equation \eqref{eqn: geodesic of mixture connection quantum case} that $\nabla^{m}$ is not complete on $\stsph$. 

In the context of Quantum Information Geometry,  and, contrary to the situation in Classical Information Geometry,  the Riemannian structure on $\stsph$ is not uniquely determined, and there is an infinite family of so-called \grit{quantum monotone metric tensors} \cfr{\cite{Petz-1996}}.
In particular, given an operator monotone function $f\colon(0,\infty)\ra(0,\infty)$  such that
\be\label{eqn: Petz function}
f(x)=xf(x^{-1}),\quad f(1)=1,
\ee
Petz's classification of quantum monotone metric tensor assures us that every such metric tensor $G_{f}$ is associated with $f$ by means of
\be\label{eqn: Petz metric 1}
\left(G_{f}(V,W)\right)(\rho)=\Tr \left(V_{\rho}\,T_{\rho}^{f} (W_{\rho})\right),
\ee
where $V,W$ are vector fields on $\stsph$,  and where  $T_{\rho}^{f}= (K^{f}_{\rho})^{-1}$ with 
\be\label{eqn: Petz metric 2}
K^{f}_{\rho}= f(L_{\rho}\,R_{\rho^{-1}})\,R_{ \rho },
\ee
where $L_{\rho}(\mathbf{a})=\rho\mathbf{a}$ and $R_{\rho}(\mathbf{a})=\mathbf{a}\rho$.
Note that both $L_{\rho}$ and $R_{\rho}$ are (invertible) superoperators on $\bh$.

Being self-adjoint, every $\rho$ can be diagonalized using an orthonormal basis $\{|j\rangle\}_{j=1,...,\mathrm{dim}(\hh)}$ of $\hh$.
Writing $\mathbf{e}^{\rho}_{lm}=| l \rangle\langle m|$, we have
\be
\rho=\sum_{j=1}^{n}\,p^{j}_{\rho}\,\mathbf{e}_{jj}^{\rho}.
\ee
Note that the collection of all $\mathbf{e}^{\rho}_{lm}$'s forms a basis of $\bh$.

We can also introduce the super-operators  $E_{kj}^{\rho}$ acting on $\bh$ according to
\be\label{eqn: rho super eigenprojectors}
E_{kj}^{\rho}\left( \mathbf{e}^{\rho}_{lm}\right)\,=\,\delta_{jl}\,\delta_{km}\mathbf{e}_{jk}^{\rho},
\ee
and it is then a matter of straightforward computation to check that
\be\label{eqn: K superoperator on eigenprojectors}
K^{f}_{\rho}=\sum_{j,k=1}^{n}\,p^{k}_{\rho}\,f\left(\frac{p^{j}_{\rho}}{p^{k}_{\rho}}\right)\,E_{kj}^{\rho} ,
\ee
and
\be\label{eqn: eigendecomposition of Petz superoperator}
T^{f}_{\rho}=\sum_{j,k=1}^{n}\,\left(p^{k}_{\rho}\,f\left(\frac{p^{j}_{\rho}}{p^{k}_{\rho}}\right)\right)^{-1}\,E_{kj}^{\rho}.
\ee
As a side remark, note that, for every vector field  $W$ such that its associated tangent vector $W_{\rho}$ at $\rho$ commutes with $\rho$, and for every vector field $V$, using equations \eqref{eqn: eigendecomposition of Petz superoperator}  and \eqref{eqn: F-R metric tensor} it follows that 
\be\label{eqn: Petz and F-R}
\left(G_{f}(V,W)\right)(\rho)=\left(G_{FR}(\vec{V}_{d},\vec{W}_{d})\right)(\vec{p}_{\rho})
\ee
with $\vec{p}_{\rho}=(p^{1}_{\rho},...,p^{n}_{\rho})$, and $\vec{V}_{d}(\vec{p}_{\rho}), \vec{W}_{d}(\vec{p}_{\rho})$ the vectors built out of the diagonal elements of $V_{\rho}$ and $W_{\rho}$ with respect to the basis of eigenvectors of $\rho$.
Note that equation \eqref{eqn: Petz and F-R} holds for every choice of the operator monotone function $f$.

From  proposition  \ref{prop: teleparallel dual connection} we know that the $G_{f}$-dual connection $\nabla^{f}$ of the mixture connection $\nabla^{m}$ is the teleparallel connection associated with the gradient vector fields of an almost dual basis of $\{L_{k}\}_{k=1,...,n}$.
An almost dual basis $\{\theta^{k}\}_{k=1,...,n^{2}-1}$ of  $\{L_{k}\}_{k=1,...,n^{2}-1}$ is then easily seen to be associated with a choice of basis $\{\omega_{k}\}_{k=0,...,n^{2}-1}$ for $\bhsa$ in such a way that $\omega_{0}=\mathbb{I}$ and  $\Tr(\omega_{k})=0$.
Indeed, setting $\theta^{k}:= \mathrm{d}l^{k}$ with 
\be
l^{k}(\rho):= \Tr(\rho\omega_{k}) 
\ee
the   expectation-value function associated with $\sigma_{k}$, a direct computation shows that
\be\label{eqn: almost dual basis quantum}
\theta^{k}(L_{j})=\frac{\mathrm{d}}{\mathrm{d}t}\left(l^{k}(\rho_{j}(t))\right)_{t=0}\stackrel{\mbox{\eqref{eqn: geodesic of mixture connection quantum case}}}{=}\Tr(\sigma_{j}\omega_{k}),
\ee
and the right-hand-side does not depend on the point $\rho$.
Then, we have
\be
\Tr(\sigma_{j}\omega_{k})\stackrel{\mbox{\eqref{eqn: almost dual basis quantum}}}{=}\left(\theta^{k}(L_{j})\right)(\rho)\stackrel{\mbox{\eqref{eqn: gradient vector field}}}{=}\left(G_{f}(L_{j},Y^{f}_{k})\right)(\rho)\stackrel{\mbox{\eqref{eqn: Petz metric 1}\eqref{eqn: linear vector field quantum}}}{=}\Tr \left(\sigma_{j}\,T_{\rho}^{f} (Y^{f}_{k}(\rho))\right)
\ee
for all $k=1,...,(n^{2}-1)$, from which it follows that the gradient vector field $Y^{f}_{k}$ associated with $\theta^{k}$ by means of $G_{f}$ must be such that
\be
\omega_{j}- T_{\rho}^{f} (Y^{f}_{k}(\rho))= c_{\rho}\mathbb{I}\equiv c_{\rho}\omega_{0}.
\ee
Therefore, since $T_{\rho}^{f}= (K^{f}_{\rho})^{-1}$, we have 
\be
Y_{k}^{f}(\rho)=K^{f}_{\rho}(\omega_{k} - c_{\rho} \mathbb{I})\stackrel{\mbox{\eqref{eqn: K superoperator on eigenprojectors}}}{=}K^{f}_{\rho}(\omega_{k} ) - c_{\rho}\rho,
\ee
with $c_{\rho}=\Tr(K^{f}_{\rho}(\omega_{k} ) )$ because it must be $\Tr(Y_{k}^{f}(\rho))=0$ since $Y_{k}^{f}(\rho)\in T_{\rho}\stsph$.

From equation \eqref{eqn: torsion teleparallel connection} we know that the torsion tensor of a teleparallel connection is connected with the commutator of the basis vector fields defining the teleparallel connection.
It is then relevant to review the results presented in \cite{Ciaglia-2020,C-DN-2021,C-DN-2022} in relation with the teleparallel connection $\nabla^{f}$  determined by $\{Y^{f}_{k}\}_{k=1,...,n}$.
In the above-mentioned works, even if their role  with respect to $G_{f}$-dual pairs was not known, the gradient vector fields associated with the expectation value functions by means of monotone quantum metric tensors have been studied in relation with the actions of Lie groups suitably extending the unitary group $\Uh$ and its canonical action 
\be\label{eqn: unitary group action}
\rho\mapsto \alpha(\mathbf{U},\rho)=\mathbf{U}\rho\mathbf{U}^{\dagger}
\ee
on $\stsph$, and a great deal of attention has been given to the commutators between these vector fields.

When the operator monotone function is  $f(x)=\kappa(x-1)/\ln(x)$ with $\kappa>0$, the Riemannian metric tensor $G_{f}$ reduces to the Bogoliubov-Kubo-Mori metric tensor \cfr{\cite{Petz-1996}} and the gradient vector fields mutually commute.
Consequently, the teleparallel connection $\nabla^{f}$ is torsion-free and flat.
This is not surprising because it is well-known  that the connection which is dual to $\nabla^{m}$ with respect to the Bogoliubov-Kubo-Mori metric tensor is the quantum version of the exponential connection which is torsion-free, flat, and also complete \cfr{\cite{F-M-A-2019,G-R-2001,Nagaoka-1995}}.
What is surprising is the following fact (we refer to \cite{Ciaglia-2020,C-DN-2021,C-DN-2022} for more details on the construction we are about to resume).
Together with  the fundamental vector fields of the standard action of $\Uh$ on $\stsph$ \cfr{equation \eqref{eqn: unitary group action}}, the gradient vector fields giving rise to connection $\nabla^{f}$ can be used to build an anti-representation of the Lie algebra of the cotangent group $T^{*}\Uh$.
This anti-representation integrates to a transitive and smooth left action on $\stsph$ which may be thought of as the quantum counterpart of the action $\gamma$ \cfr{equation \eqref{eqn: geodesics of classical exponential connection}} of $\mathbb{R}^{n}_{+}$ on $\Delta_{n}$.
In particular, every geodesic of $\nabla^{f}$ starting at $\rho$ with initial velocity $\mathbf{v}$ is of the form
\be\label{eqn: geodesics quantum exponential connection}
\rho_{\mathbf{v}}(t)=\frac{\mathrm{e}^{\ln\rho + t\mathbf{v}}}{\Tr\left(\mathrm{e}^{\ln\rho + t\mathbf{v}}\right)}
\ee
from which it clearly follows that $\nabla^{f}$ is complete on $\stsph$.
Moreover, when $[\mathbf{v},\rho]=\mathbf{0}$, the geodesic in equation \eqref{eqn: geodesics quantum exponential connection}, if diagonalized, has basically the form of the geodesic of the classical exponential connection given in  equation  \eqref{eqn: geodesics of classical exponential connection 2}.
Since both $\nabla^{m}$ and $\nabla^{f}$ are torsion-free in this case, the covariant tensor $T$ is completely symmetric and $(\stsph,G_{f},T)$ is a Statistical Manifold in the sense of \cite{Lauritzen-1987}.

It is also worth noting that, when $f(x)=\frac{\kappa(x-1)(x^{\kappa}+1)}{2(x^{\kappa}-1)}$ with $0<\kappa<1$, a similar instance manifests itself.
Indeed, again together with the fundamental vector fields of the standard action of $\Uh$ on $\stsph$ \cfr{equation \eqref{eqn: unitary group action}}, the gradient vector fields giving rise to connection $\nabla^{f}$ provide a Lie algebra anti-representation.
However, this time the anti-representation is of the Lie algebra $\mathfrak{gl}(\hh)$ of the general linear group $\Glh$.
This representation  integrates to a transitive and smooth left action of $\Glh$ on $\stsph$.
In particular, every geodesic of $\nabla^{f}$ starting at $\rho$ and with initial tangent vector $\mathbf{v}=v^{j}\sigma_{j}$, where $v^{j}\in\mathbb{R}$ for every $j=1,...,(n^{2}-1)$, is of the form
\be\label{eqn: geodesic of deformed bures connection}
\rho_{\mathbf{v}}(t)=\frac{\left(\mathrm{e}^{t\mathbf{v}}\rho^{\sqrt{\kappa}}\mathrm{e}^{t\mathbf{v}}\right)^{\frac{1}{\sqrt{\kappa}}}}{\Tr\left(\left(\mathrm{e}^{t\mathbf{v}}\rho^{\sqrt{\kappa}}\mathrm{e}^{t\mathbf{v}}\right)^{\frac{1}{\sqrt{\kappa}}}\right)},
\ee
from which it is clear that $\nabla^{f}$ is always complete.

When $\kappa=1$, the monotone quantum metric tensor $G_{f}$ coincides with the Bures-Helstrom metric tensor \cfr{\cite{B-Z-2006,C-J-S-2020,Dittmann-1995,Helstrom-1967,Uhlmann-1992,Safranek-2018}}  ubiquitous in quantum estimation theory \cfr{\cite{C-J-S-2020-02,Fuchs-1996,L-Y-L-W-2020,Paris-2009,Suzuki-2019,S-Y-H-2020,T-A-D-2020}}, while when $\kappa=\frac{1}{4}$, the monotone quantum metric tensor $G_{f}$ coincides with the Wigner-Yanase metric tensor \cite{G-I-2001,G-I-2003,Hasegawa-2003,Jencova-2003-2}.
In all these cases, the commutator  between gradient vector fields is a fundamental vector field of $\alpha$ and, in general, it does not vanish  so that the connection $\nabla^{f}$ presents torsion.
Therefore, the tensor $T$ determined by $\nabla^{m},\nabla^{f}$, and $G_{f}$ will not be symmetric in the first and second entries so that $(\stsph,G_{f},T)$ is not a Statistical Manifold in the sense of \cite{Lauritzen-1987}.

\begin{remark}
	Incidentally, for a two-level quantum system, the cases $f(x)=\kappa(x-1)/\ln(x)$ with $\kappa>0$ and $f(x)=\frac{\kappa(x-1)(x^{\kappa}+1)}{2(x^{\kappa}-1)}$ with $0<\kappa<1$ are the only cases of monotone quantum metric  tensors  for which the gradient vector fields and the fundamental vector fields of the standard action of $\Uh$ on $\stsph$ give rise to a transitive group action on $\stsph$ \cfr{\cite{C-DN-2021}}.
\end{remark}

\section{Conclusions} \label{sec: conclusions}

In this work we investigated the notion of $G$-dual teleparallel  pairs  $(\nabla,\nabla^{*})$ of affine connections on a real, finite-dimensional, smooth, parallelizable Riemannian manifold $(\mathcal{N},G)$.
We showed how every such pair gives rise to a covariant $(0,3)$ tensor $T$ generalizing the Amari-\v{C}encov tensor of Statistical Manifolds to the case where connections present torsion.
Then, we reviewed some examples of $G$-dual teleparallel pairs in the context of Classical and Quantum Information Geometry, where the existence of a suitable (torsion-free and flat) mixture connection $\nabla^{m}$ (associated with a suitable convex structure) gifts us with a preferred teleparallel connection.

In the classical case, it turns out that the dual connection is also torsion-free and flat, and it is known as  the  exponential connection, a fact which is well-known. 
Following the theory developed in section \ref{sec: dually-weitzenbock manifolds}, we find the explicit form of a basis of vector fields determining the exponential connection as the associated teleparallel connection, and show how these vector fields provide a homogeneous manifold structure for the interior of the n-simplex $\Delta_{n}$ such that the geodesics of the exponential  connection  coincides with the images of one-parameter subgroups.
As far as we know, these geometric considerations regarding the exponential connection are new.

In the quantum case, we investigate $G_{f}$-dual pairs $(\nabla,\nabla^{*})$ where $G_{f}$ is a monotone quantum metric tensor, and $\nabla$ is the mixture connection $\nabla^{m}$ on the manifold $\stsph$ of faithful quantum states of a quantum system with finite-dimensional Hilbert space $\hh$.
The theory developed in section \ref{sec: dually-weitzenbock manifolds} allows us to conclude that $\nabla^{*}\equiv\nabla^{f}$ is the teleparallel connection associated with the gradient vector fields of the expectation value functions.
These gradient vector fields already appeared in the literature in connection with transitive actions on $\stsph$ of suitable extensions of the unitary group $\Uh$, and we are here able to make a point of contact between these group actions and the properties of $\nabla^{f}$.
In particular, whenever these groups appear, we are able to explicitly write down the geodesics of $\nabla^{f}$ again as the images of one-parameter subgroups, from which it follows that $\nabla^{f}$ is a complete connection.
However, it is in general true that $\nabla^{f}$ presents torsion unless $G_{f}$ is the so-called Bogoliubov-Kubo-Mori metric tensor, in which case $\nabla^{f}$ is the quantum counterpart of the exponential connection.

A natural evolution of the ideas presented here would be the investigation of how $G$-dual teleparallel pairs “projects down” on homogeneous manifolds.
Indeed, it is clear that $G$-dual teleparallel pairs only exist on parallelizable manifolds, but not all manifolds of interest in Classical and Quantum Information Geometry are parallelizable.
A very relevant example is given by the manifold of pure states of a finite-level quantum system, which is diffeomorphic with the complex projective space  $\mathbb{P}(\hh)$.
In this and other cases, however, the manifold of interest $\mathcal{N}$ is actually a homogeneous manifold for some Lie group $\gapp$ (e.g., the unitary group $\Uh$ for the manifold of pure quantum states), and Lie groups are obviously parallelizable manifolds.
Then, we may pull-back the Riemannian metric tensor on $\mathcal{N}$  “coming from” Information Geometry to the Lie group $\gapp$ by means of the canonical projection  $\pi\colon \gapp\ra\mathcal{N}$, and complement it to a Riemannian metric tensor $G$ on $\gapp$ for which  $G$-dual teleparallel pairs exist.
At this point, $\pi$ becomes a Riemannian submersion and the idea would be to  understand if and how it is possible to “project down” a $G$-dual teleparallel pair on the homogeneous manifold in a suitable way, perhaps exploiting a suitable projector to describe the module of vector fields on $\mathcal{N}$ in terms of the module of vector fields on $\gapp$. 
We are currently investigating this problem,  and we plan to discuss it in a future publication.

\section*{Funding}

This work has been supported by the Madrid Government (Comunidad de Madrid-Spain) under the Multiannual Agreement with UC3M in the line of “Research Funds for Beatriz Galindo Fellowships” (C\&QIG-BG-CM-UC3M), and in the context of the V PRICIT (Regional Programme of Research and Technological Innovation).
The authors acknowledge financial support from the Spanish Ministry of Economy and Competitiveness, through the Severo Ochoa Programme for Centres of Excellence in RD (SEV-2015/0554), the MINECO research project  PID2020-117477GB-I00,  and Comunidad de Madrid project QUITEMAD++, S2018/TCS-A4342.
G. M. is also a member of the Gruppo Nazionale di Fisica Matematica (INDAM), Italy. 
F. D. C. acknowledges support from the CONEX-Plus programme funded by Universidad Carlos III de Madrid and the European Union’s Horizon 2020 research and innovation programme under the Marie Sklodowska-Curie grant agreement No. 801538.

\addcontentsline{toc}{section}{References}

\end{document}